%% file: main.tex
\makeatletter
\disable@package@load{program}{}
\makeatother

\documentclass[iicol,pdflatex,sn-mathphys-num]{sn-jnl}

\usepackage{graphicx}%
\usepackage{multirow}%
\usepackage{amsmath,amssymb,amsfonts}%
\usepackage{amsthm}%
\usepackage{mathrsfs}%
\usepackage[title]{appendix}%
\usepackage{xcolor}%
\usepackage{textcomp}%
\usepackage{manyfoot}%
\usepackage{booktabs}%
\usepackage{algorithm}%
\usepackage{algorithmicx}%
\usepackage{algpseudocode}%
\usepackage{listings}%
\usepackage{comment}

\usepackage{url} 
\usepackage{pifont} 
\usepackage{bm}
\usepackage{array}

\usepackage{hyperref}

\let\orcidlogo\relax
\usepackage{orcidlink}
\renewcommand{\orcid}[1]{}

\theoremstyle{thmstyleone}%
\newtheorem{theorem}{Theorem}%
\newtheorem{proposition}[theorem]{Proposition}%

\theoremstyle{thmstyletwo}%
\newtheorem{remark}{Remark}%

\theoremstyle{thmstylethree}%

\definecolor{darkgreen}{RGB}{0,128,0}
\newcommand{\rev}[1]{{\color{black}#1}} 
\newcommand{\revtwo}[1]{{\color{black}#1}} 

\usepackage{enumitem}
\setlist{nosep, leftmargin=*} 

\begin{document}

\title[]{eIRWR: Enhanced Iterative Random Walk with Restart for Scalable Root Cause Analysis in Microservices}

\author*[1]{\fnm{Saiful} \sur{Khan}\orcidlink{0000-0002-6796-5670}}\email{saiful.khan@stfc.ac.uk}
\author[2]{\fnm{Afrah} \sur{Farea}\orcidlink{0000-0003-4412-5377}}

\affil*[1]{\orgdiv{Science and Technology Facilities Council (STFC)}, \orgname{UKRI}, \orgaddress{\postcode{OX11 0QX}, \country{UK}}}

\affil[2]{\orgdiv{Department of Engineering Science}, \orgname{University of Oxford}, \orgaddress{\street{Parks Road}, \postcode{OX1 3PJ}, \country{UK}}}

\abstract{
Root cause analysis (RCA) in microservice systems must find the faulty service behind symptoms that spread across hundreds or thousands of dependent services.
Graph-based random walk methods propagate anomaly evidence over the service dependency graph.
However, existing anomaly-restart walks restart from the raw anomaly vector, which is \rev{often} dominated by \rev{strongly anomalous} downstream victims rather than the \rev{less anomalous} source.
Through a controlled ablation of a basic anomaly-restart walk, the Iterative Random Walk with Restart (IRWR), we show that its ``resilience damping'' is equivalent to a higher restart probability in Personalized PageRank (PPR). Basic IRWR is therefore no stronger than a restart-tuned PPR, which we adopt as a fairer baseline.
We then present Enhanced IRWR (eIRWR), which goes beyond restart tuning. It (a)~concentrates restart mass on the most suspicious nodes through power-law sharpening and (b)~adds self-loops and backward edges so that probability accumulates at cascade sources.
On three large topologies (12K--25K nodes) from the Alibaba Microservice Trace Dataset, eIRWR \rev{is the most accurate of the compared graph-based methods at every level of root-cause visibility. At visibility $v = 0.3$, it reaches} a Mean Reciprocal Rank (MRR) of \rev{$0.78$, about four times a restart-tuned PPR and nine times the best other baseline}. At high visibility, it reaches MRR\,$= \rev{0.95}$ and runs in under \rev{30}\,ms on graphs with 17K nodes, so it is suitable for online use.
}


\maketitle

\input{sections}

\rev{\bmhead{Code availability}
The source code is available at \url{https://github.com/saifulkhan/eirwr-rca}.}

\bibliography{references}


\end{document}

%% file: sections.tex
\section{Introduction}
\label{sec:introduction}

\begin{figure*}
    \centering
    \includegraphics[width=0.85\linewidth]{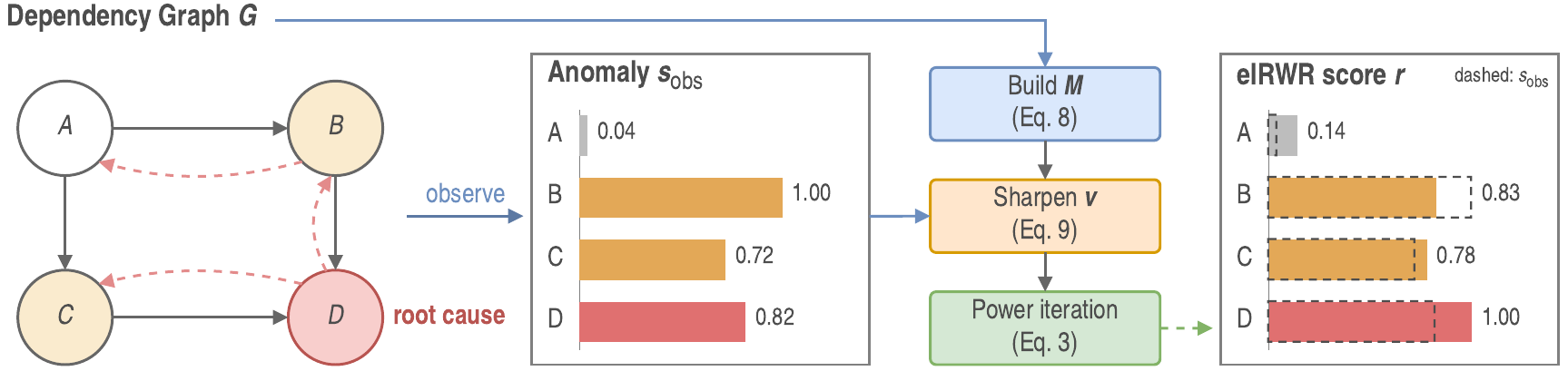} 
    \caption{
      \rev{The} eIRWR pipeline \rev{on four services of a simulated incident on the CG84 topology (Section~\ref{sec:setup})}. Service $D$ (root cause) fails and cascades symptoms backward to its callers $B$ and $C$ (dashed arrows). The observed anomaly vector $\bm{s}_{\mathrm{obs}}$ shows $B$, a victim, as \rev{the service with the highest anomaly score}. eIRWR builds the transition matrix, sharpens the restart (teleportation) vector, and runs power iteration \rev{on the full graph. The right panel shows the resulting scores, which rank the root cause~$D$ first. Bars show each score divided by the largest score among all 11,966 services (visibility $v = 0.3$). Bar colors mark the root cause $D$ (red), its direct callers $B$ and $C$ (amber), and the indirect caller $A$ (grey). Dashed outlines in the right panel show $\bm{s}_{\mathrm{obs}}$ for comparison.}
      }
    \label{fig:irwr-overview}
\end{figure*}

Modern cloud-native applications split monolithic systems into many loosely coupled microservices that communicate over the network, following the principles of service-oriented computing (SOC)~\cite{Monteiro2020,wang2023soca,Khan:2022:TSC,Khan:2022:JDIM,Khan:2025:SCC,Khan:2023:IEEEAccess}.
This architecture brings agility and independent scaling, but it also creates complex failure modes. A fault in one service can cascade through chains of synchronous calls and cause latency spikes, error bursts, and poor user experience in seemingly unrelated services. Finding the root cause (the originating faulty service) quickly is critical to reduce downtime.

During an incident, monitoring and anomaly-detection systems~\cite{khedimi2022soca, ahmadpanah2026soca} observe anomaly signals (e.g., higher latency or error rates) at many services at once.
The root cause analysis (RCA) task is to rank services by how likely each is to be the originating fault. The inputs are (i)~the service dependency graph extracted from distributed traces and (ii)~the anomaly vector aggregated at the service level.
The difficulty is that the root cause often does not have the highest anomaly score. Downstream victims of the cascade frequently \rev{have higher anomaly scores than the source, for instance larger latency increases or error rates. For example, a caller that retries a slow request adds its own waiting time and retry errors to the delay of the faulty service, so its anomaly score can exceed that of the root cause.}

We focus on the aggregate-metric setting. No per-request traces are available, only a service-level call graph and an anomaly vector derived from aggregate metrics (e.g., P99 latency, error rate). This setting applies most widely. Some systems support per-request tracing~\cite{yu2021microrank, Yu:2023:J.Softw.Evol.Process}, but many production deployments export only aggregate metrics because of overhead, sampling, or privacy constraints. \rev{eIRWR also needs the call graph. When no call graph is available, ranking services by their anomaly score, or another method that needs no graph, should be used instead.}

Graph-based random walk methods, such as Personalized PageRank (PPR)~\cite{Wu:2020:NOMS}, CloudRanger~\cite{Wang:2018:CCGRID}, and MonitorRank~\cite{Kim:2013:SIGMETRICS}, are a principled approach to aggregate-metric RCA. They propagate anomaly evidence through the dependency graph to localize the root cause. However, they walk on the dependency weight matrix~$W$ of normalized call frequencies. This matrix captures which services communicate, but not how faults propagate. It does not model the direction of fault cascades (backward, from callees to callers). It also ignores that resilient services \rev{can contain} failures without passing them on. \rev{A service contains a failure (fault containment~\cite{Avizienis:2004:TDSC}) when it handles the fault locally, so that its own callers see little or no effect. For example, if a recommendation service times out, its caller can use a circuit breaker to return a cached or default response, and the user request still succeeds.} As a result, the walker can spread probability to irrelevant parts of the graph, especially when the root cause has low visibility in the anomaly vector.

\rev{Fig.~\ref{fig:irwr-overview} shows four services from a simulated incident on one of our evaluation topologies (CG84, Section~\ref{sec:setup}). Service~$A$ calls $B$ and $C$, and both $B$ and $C$ call $D$. A fault at $D$ (the root cause) propagates backward to its callers $B$ and $C$. In the observed anomaly vector, $B$ has the highest score of all 11,966 services, and $D$ reaches only $0.82$ of it. $B$ sends $62\%$ of its calls to $D$, so it inherits much of the fault, while the fault is only weakly visible at $D$ itself. Ranking services by anomaly score therefore blames the victim~$B$. A random walk seeded with the anomaly vector has the same bias, because each restart re-injects most of the probability mass at $B$. PPR~\cite{Wu:2020:NOMS}, for example, ranks $D$ only 86th. This example shows what an RCA walk must do: restart on the most suspicious nodes, move from victims back to the services they call, and let probability accumulate at the source. Section~\ref{sec:worked-example} follows this example through each mechanism of eIRWR.}

We study the anomaly-seeded walk over the dependency graph, called the Iterative Random Walk with Restart (IRWR), through a controlled ablation.
Two mechanisms often credited for accuracy gains, transposing the transition matrix and damping it with a uniform ``resilience'' factor, do not explain the improvement.

The real improvement comes from reshaping how the walk restarts and where it focuses. Our method, Enhanced IRWR (eIRWR), makes two main changes:
(i)~It sharpens the restart (teleportation) distribution with a power law, so that restart mass concentrates on the most suspicious nodes.
(ii)~It adds self-loops and backward edges to the transition matrix, so that probability accumulates at cascade sources.

Our contributions are:

\begin{enumerate}
    \item[(a)]
    We show that the accuracy gain of resilience-damped walks over PPR is fully explained by a higher effective restart probability, not by directionality or resilience heterogeneity. We therefore evaluate against a restart-tuned PPR, a fairer and stronger baseline than the default configuration used in prior work.

    \item[(b)]
    We propose eIRWR, which combines the two changes above.

    \item[(c)]
    We evaluate eIRWR on three large topologies (12K--25K nodes) from the Alibaba Microservice Trace Dataset~\cite{alibaba2022trace} against eight \rev{graph-based} baselines.
\end{enumerate}

\section{Related Work}\label{sec:related-work}

RCA and anomaly detection in microservices are widely studied; see~\cite{soldani2022survey} for a survey.
Simple metric-based RCA methods fall into two categories: direct metric ranking and graph-aware analysis. Direct ranking orders services by anomaly scores, such as latency spikes or error rates. It often blames downstream victims, especially when the root cause has low visibility.

\vspace{1mm}\noindent\textbf{Graph-Based Random Walk.}
Graph-based methods model service dependencies as a directed graph and run random walks to propagate anomaly evidence.
PageRank~\cite{brin1998pagerank} ranks web pages by structural importance. For RCA, it finds structurally important services but ignores anomalies, so it gives the same ranking for every incident.
MicroRCA~\cite{Wu:2020:NOMS} adapts PPR by replacing the uniform teleportation vector with the observed anomaly vector~$\bm{s}_{\mathrm{obs}}$, so the walker restarts mostly at anomalous nodes. The iteration is $\bm{r}^{(k+1)} = (1-\alpha)\,\bm{W}\,\bm{r}^{(k)} + \alpha\,\bm{s}_{\mathrm{obs}}$,
where~$\bm{W}$ is the row-stochastic dependency weight matrix.
MicroRCA is the closest prior work to eIRWR, which uses the same power-iteration framework.
However, MicroRCA seeds the walk with raw anomaly scores at a fixed restart rate. eIRWR instead sharpens the restart vector and augments the transition matrix to concentrate probability at cascade sources.

CloudRanger~\cite{Wang:2018:CCGRID} uses a second-order random walk in which transitions depend on both edge weights and neighbor anomaly scores. The probability of moving from node \( i \) to neighbor \( j \) is proportional to \( w_{ij} \cdot s_{\mathrm{obs}}(j) \), which steers the walker toward anomalous regions. However, it relies on many short stochastic walks. This can make it less efficient than deterministic power iteration and sensitive to the number of walks. MonitorRank~\cite{Kim:2013:SIGMETRICS} instead reweights the transition matrix by the anomaly correlation between connected services, \( \bm{W}_{\mathrm{corr}} = \mathrm{diag}(\sqrt{\bm{s}})\,\bm{W}\,\mathrm{diag}(\sqrt{\bm{s}}) \). This keeps the matrix sparse. \revtwo{After row normalization, the left factor cancels, so transitions favor anomalous callees.} MonitorRank lets the anomaly context shape the walk, but it does not model how faults propagate.

\vspace{1mm}\noindent\textbf{Graph Traversal and Pruning.}
These methods traverse the call graph for a fixed number of hops instead of iterating to convergence. Alibaba's MicroHECL~\cite{Liu:2021:ICSE-SEIP} traces symptoms backward through the call graph using transposed weights and accumulates evidence at each hop up to a fixed depth. It is efficient, with cost $O(d\,|E|)$. However, the fixed depth limits long-range dependencies. It also lacks the restart of RWR-based methods, so evidence can drift into irrelevant regions. Microsoft's TraceDiag~\cite{Ding:2023:ACM} prunes low-anomaly nodes to reduce noise and then runs PPR on the pruned subgraph. This improves scalability in noisy environments. However, a hard pruning threshold may discard the true root cause when its visibility is low, which is often when RCA is most needed.

\vspace{1mm}\noindent\textbf{Spectrum-Based and Trace-Level.}
These methods use per-trace data rather than aggregate anomaly vectors.
MicroRank~\cite{yu2021microrank} combines Spectrum-Based Fault Localization (SBFL) with extended PageRank. The Ochiai formula scores each service~$i$ by its participation in failing and passing traces, $\mathit{Score}_i = e_f^{(i)} / \sqrt{(e_f^{(i)} + e_p^{(i)})(e_f^{(i)} + n_f^{(i)})}$.
These scores serve as teleportation weights in PageRank, and the final ranking combines the spectrum and PageRank scores.

TraceRank~\cite{Yu:2023:J.Softw.Evol.Process} extends MicroRank in three ways: (i)~a score $\mathrm{corr}_i = |\mathrm{PCC}(t_i, t_{\mathrm{frontend}})|$, the absolute Pearson correlation between each service's processing time and the front-end latency; (ii)~an adjacency matrix with forward, backward, and self-loop transitions weighted by these scores; and (iii)~a correction step that promotes services that are both spectrally suspicious and statistically abnormal. TraceRank performs well but needs per-trace data, which is not always available at scale.
Our method needs only the aggregate anomaly vector and the dependency graph.

Neighbor anomaly correlation methods, such as PAL~\cite{pal2020rca} and FChain~\cite{fchain2022}, rank services by how their anomaly correlates with that of their neighbors. A common form combines caller anomaly (the cascade signal), own anomaly, and callee anomaly. These methods are fast but use only first-order neighbors, so they miss multi-hop propagation.

\vspace{1mm}\noindent\textbf{Fault Management in SOC.}
Fault management also has a long history in service-oriented computing.
Model-based approaches diagnose faults from behavioral models of composed services, for example with partially observed stochastic Petri nets~\cite{bhandari2020soca}. Graph-matching approaches compare an anomalous system state with a library of known anomalous graphs~\cite{brandon2020jss}.
Data-driven pipelines detect anomalies in service platforms~\cite{khedimi2022soca} or combine detection with autonomous in-kernel remediation for cloud-native systems~\cite{ahmadpanah2026soca}. Deep learning has also been applied to performance diagnosis in cloud microservices~\cite{wu2020icsoc}.
These approaches need behavioral models, curated anomaly libraries, or trained models.
eIRWR complements them. It uses only the dependency graph and aggregate metrics that service-oriented deployments already expose, needs no training, and localizes the faulty service rather than only flagging the anomaly.

In summary, eIRWR keeps the efficiency and iterative convergence of power-iteration walks such as PPR (MicroRCA). Like them, it uses aggregate metrics without per-trace data (Table~\ref{tab:method-comparison}).
What sets eIRWR apart is not a new transition matrix, because we show that the resilience-damped transition reduces to an effective restart. It is a shaped restart: power-law teleportation sharpening, combined with self-loops and backward edges that concentrate probability at cascade sources. Unlike TraceRank and MicroRank, it needs no per-trace data. Unlike MicroHECL, it converges to a stationary distribution instead of stopping at a fixed depth.

\begin{table}[t]
\centering
\caption{Comparison of RCA methods. \textbf{Sharpen}: shaped restart (sharpening, self-loops, or backward edges); \textbf{Iter.}: iterative convergence; \textbf{Trace}: needs per-trace data; \textbf{Weighted}: call-frequency weights.}
\label{tab:method-comparison}

\setlength{\tabcolsep}{3pt}
\begin{tabular*}{\columnwidth}{@{\extracolsep{\fill}}lcccc@{}}
    \toprule
    \textbf{Method}       & \textbf{Sharpen} & \textbf{Iter.} & \textbf{Trace} & \textbf{Weighted} \\
    \midrule
    PageRank               & --         & \checkmark & --         & \checkmark \\
    MicroRCA (PPR)         & --         & \checkmark & --         & \checkmark \\
    CloudRanger            & --         & --         & --         & \checkmark \\
    MonitorRank            & --         & \checkmark & --         & \checkmark \\
    In-Degree              & --         & --         & --         & \checkmark \\
    MicroHECL              & --         & --         & --         & \checkmark \\
    TraceDiag              & --         & \checkmark & --         & \checkmark \\
    Neigh.\ Corr.          & --         & --         & --         & \checkmark \\
    MicroRank              & --         & \checkmark & \checkmark & \checkmark \\
    TraceRank              & \checkmark & \checkmark & \checkmark & --         \\
    \midrule
    IRWR (baseline)        & --   & \checkmark & --         & \checkmark \\
    \textbf{eIRWR (ours)} & \checkmark & \checkmark & --         & \checkmark \\
    \bottomrule
\end{tabular*}
\end{table}

\section{Proposed Method}
\label{sec:method}

We present our method in two stages: (a)~a basic anomaly-restart walk (IRWR), used mainly to build a fair, restart-tuned baseline, and (b)~the enhanced method (eIRWR).
We distinguish three terms: the restart rate (the scalar probability~$\alpha$), the restart distribution (the teleportation vector the walk restarts into), and restart shaping (changing that distribution).

\subsection{Problem Formulation}\label{sec:problem}

Consider a microservice system with $N$ services. We model it as a directed weighted graph $G = (V, E, \bm{W})$, where:
\begin{enumerate}
  \item[-] $V = \{v_1, \ldots, v_N\}$ is the set of microservices,
  \item[-] $E \subseteq V \times V$ is the set of directed edges (service $v_i$ calls service $v_j$ if $(v_i, v_j) \in E$),
  \item[-] $\bm{W} \in \mathbb{R}^{N \times N}$ is the row-stochastic weight matrix encoding normalized request rates.
\end{enumerate}
During an incident, a monitoring system observes an anomaly vector $\bm{s}_{\mathrm{obs}} \in \mathbb{R}^N_{\geq 0}$ with $\|\bm{s}_{\mathrm{obs}}\|_1 = 1$, where $s_i$ quantifies the anomaly severity at service~$v_i$ (e.g., normalized latency deviation).
The RCA task is to rank services by their likelihood of being the root cause, that is, the originating service whose failure cascades through the dependency graph and produces the observed symptoms.

\subsection{Graph Construction}\label{sec:graph}

\noindent\textbf{Dependency weight matrix.}
We extract the dependency graph from distributed tracing data (e.g., the Alibaba Microservice Trace Dataset~\cite{alibaba2022trace}). Each edge $(v_i, v_j)$ carries a request count $\lambda_{ij}$ (the number of calls from $v_i$ to $v_j$ observed in the trace window). We build the row-stochastic weight matrix~$\bm{W}$:
\begin{equation}\label{eq:weight}
  w_{ij} = \frac{\lambda_{ij}}{\sum_{k=1}^{N} \lambda_{ik}}\,,
\end{equation}
so that each row sums to one: $\sum_j w_{ij} = 1$. Entry $w_{ij}$ is the fraction of service~$v_i$'s outgoing traffic sent to service~$v_j$. High-traffic edges carry larger weights, because services that communicate more propagate faults more readily.

\vspace{1mm}
\noindent\textbf{Resilience-damped transition matrix.}
As a modeling device, we can assume that dependencies do not all propagate failures equally. We then introduce a per-service resilience factor $R_i \in [0,1]$: the probability that service~$v_i$ \rev{contains} a failure without propagating it\rev{, for example through timeouts, retries, circuit breakers, or fallback responses (see the example in Section~\ref{sec:introduction})}.  Damping the transition weights by $(1 - R_i)$ gives the matrix
\begin{equation}\label{eq:fault-prop}
  \bm{M}_R = \mathrm{diag}(\bm{1} - \bm{R})\,\bm{W}\,,
\end{equation}
which is used in the same power iteration as PPR.
When $\bm{R}$ is uniform ($R_i \equiv R$), Eq.~\eqref{eq:fault-prop} simply scales every transition by $(1-R)$, so the transition matrix becomes sub-stochastic. Section~\ref{sec:vs-ppr} shows that this is equivalent to running PPR with a higher restart probability. We keep Eq.~\eqref{eq:fault-prop} because basic IRWR and several prior walks use it, but we compare against a restart-tuned PPR that already has this effect.

\subsection{Basic IRWR}\label{sec:irwr-basic}

Basic IRWR computes a root-cause probability vector $\bm{r} \in \mathbb{R}^N$ by anomaly-seeded power iteration on $\bm{M}_R$ (Eq.~\eqref{eq:fault-prop}):
\begin{equation}\label{eq:rwr}
  \bm{r}^{(k+1)}
    = (1 - \alpha)\,\bm{M}_R\,\bm{r}^{(k)}
    + \alpha\,\bm{s}_{\mathrm{obs}}\,,
  \quad k = 0, 1, 2, \ldots
\end{equation}
where $\alpha \in (0,1)$ is the restart probability and the initial state is $\bm{r}^{(0)} = \bm{s}_{\mathrm{obs}}$. Because $\bm{W}$ is row-stochastic over out-edges, one application of $\bm{M}_R$ adds each callee's score to its callers, weighted by traffic share and damped by $(1-R_i)$. The restart term re-injects the observed anomaly at every step. Setting $\bm{R}=\bm 0$ gives exactly the PPR of MicroRCA.

\vspace{1mm}
\noindent\textbf{Convergence.}
Iteration stops when $\|\bm{r}^{(k+1)} - \bm{r}^{(k)}\|_1 < \varepsilon$ (default $\varepsilon = 10^{-6}$). Since \revtwo{each row of $\bm{M}_R$ sums to at most one} and $\alpha > 0$, the mapping is a contraction \revtwo{in the maximum norm} with rate at most $(1-\alpha)$, so the iteration converges. With sparse CSR matrices, each iteration costs $O(|E|)$ operations, and convergence typically takes 15--50 iterations.

\vspace{1mm}
\noindent\textbf{Root cause identification.}
The root cause is the node with the highest converged probability, $v^* = \arg\max_i\, r_i$.
Sorting nodes by decreasing $r_i$ gives a shortlist $v_{(1)}, v_{(2)}, \ldots$ for on-call engineers.

\subsection{Enhanced IRWR (eIRWR)}\label{sec:eIRWRnhanced}

Basic IRWR uses a fixed resilience $R$ and the raw anomaly vector as the teleportation target. We extend it with three mechanisms\rev{, whose contributions we measure in the ablation of Section~\ref{sec:ablation}}. \rev{Each mechanism addresses one requirement identified in the motivating example (Fig.~\ref{fig:irwr-overview}).}

\vspace{1mm}\noindent\textbf{Anomaly-conditioned resilience.} \label{sec:adaptive-R}
A service under anomalous load or errors has degraded resilience: its circuit breakers, timeouts, and retry budgets\rev{~\cite{heorhiadi2016gremlin}} are strained. We model this by making the resilience depend on the observed anomaly:
\begin{equation}\label{eq:adaptive-R}
  R_i = R_{\mathrm{base}} \cdot \exp\!\bigl(-\beta\, \hat{s}_i\bigr)\,,
\end{equation}
where $\hat{s}_i = s_i / \max_j s_j$ is the normalized anomaly score and $\beta \geq 0$ is a sensitivity parameter. When $\beta = 0$, all services share the base resilience $R_{\mathrm{base}}$. As $\beta$ increases, anomalous services become less resilient and propagate faults more readily. Services that are already distressed therefore contribute more strongly to the random walk. \rev{In the motivating example, the distressed services $B$ and $D$ therefore contribute more strongly than the weakly anomalous service~$A$.}

The adaptive transition matrix is
\begin{equation}\label{eq:M-base}
  \bm{M}_{\mathrm{base}}
    = \mathrm{diag}\!\bigl(\bm{1} - \bm{R}(\bm{s}_{\mathrm{obs}})\bigr)\,\bm{W}\,.
\end{equation}

\vspace{1mm}\noindent\textbf{Backward transitions and self-loops.} \label{sec:backward}
Basic IRWR walks only along existing dependency edges. In some topologies, the root cause may be structurally distant from the most symptomatic region. We add two mechanisms.

\vspace{1mm}
\noindent\emph{Backward transitions.}
For each forward edge $(v_i, v_j) \in E$ whose reverse $(v_j, v_i) \notin E$, we add a backward edge $(v_j, v_i)$ weighted by the \rev{belief score of the callee $v_j$}:
\begin{equation}\label{eq:backward}
  \bm{A}_{\mathrm{bwd}}[j, i] = \rho \cdot \rev{C_j}\,,
\end{equation}
where $\rho \in [0, 1)$ is a discount factor and \rev{$C_j$} is the belief score defined in Section~\ref{sec:belief}. Backward edges let the walker explore callers of anomalous services. This is essential when the root cause is upstream of the most visible symptoms. \rev{In the motivating example, backward edges let the walker move from the root cause $D$ to its callers $B$ and $C$, which are strongly anomalous, and from the victim $B$ to its own callers, which are barely anomalous. This difference lowers the score of $B$ relative to $D$ (Section~\ref{sec:worked-example}). Similar backward transitions appear in earlier random-walk RCA methods~\cite{Kim:2013:SIGMETRICS, Wang:2018:CCGRID, meng2020microcause}.}

\vspace{1mm}
\noindent\emph{Self-loops for suspicious nodes.}
A root cause tends to be more anomalous than its neighbors, because it is the source, not a victim. We add self-loops to capture this signature:
\begin{equation}\label{eq:self-loop}
  \bm{A}_{\mathrm{self}}[i,i]
    = \max\!\bigl(0,\; C_i - \max_{j} \tilde{M}_{ij}\bigr)\,,
\end{equation}
where \rev{$\tilde{\bm{M}} = \bm{M}_{\mathrm{base}}\,\mathrm{diag}(\bm{C}) + \bm{A}_{\mathrm{bwd}}$ holds the forward and backward weights before row normalization. The maximum is taken over all nonzero entries in row~$i$, that is, over the callees and callers of node~$i$.} A node gets a self-loop only if its belief exceeds its largest outgoing transition weight. This makes the walker linger at likely root causes. \rev{This follows the lazy random walk, in which a self-loop lets the walker stay at its current node~\cite{levin2017markov}, and the self-edges used in earlier RCA walks~\cite{Kim:2013:SIGMETRICS, Wang:2018:CCGRID}. In the motivating example, a self-loop at $D$ keeps the probability that reaches $D$ through $B \to D$ and $C \to D$ at the source.}

The combined transition matrix is:
\begin{equation}\label{eq:M-combined}
  \bm{M}
    = \mathrm{RowNorm}\!\bigl(
        \bm{M}_{\mathrm{base}} \cdot \mathrm{diag}(\bm{C})
        + \bm{A}_{\mathrm{bwd}}
        + \bm{A}_{\mathrm{self}}
      \bigr)\,,
\end{equation}
where $\mathrm{RowNorm}(\cdot)$ normalizes each row to sum to one.

\vspace{1mm}\noindent\textbf{Power-law teleportation sharpening.}\label{sec:teleport}
The standard teleportation vector $\bm{s}_{\mathrm{obs}}$ spreads probability mass across all nodes, including many with only background noise. We sharpen it by raising it to a power~$q$:
\begin{equation}\label{eq:teleport}
  \bm{v} = \frac{\bm{b}^{\,q}}{\|\bm{b}^{\,q}\|_1}\,,
  \qquad q \geq 1\,,
\end{equation}
where $\bm{b}$ is the current belief vector (Section~\ref{sec:belief}). For $q = 1$, this is standard teleportation. For $q > 1$, mass concentrates at the most anomalous nodes, and restarts to noisy background services are suppressed. \rev{Raising a distribution to a power and renormalizing is a standard way to sharpen it, for example in temperature sharpening~\cite{berthelot2019mixmatch}. Only the restart (personalization) vector changes, and this vector controls where the walk concentrates its probability~\cite{haveliwala2002topic, tong2006rwr}. In large graphs, much of the anomaly mass sits on many weakly anomalous background services. Sharpening removes most of it, so restarts concentrate on the few suspicious services around the cascade, such as $B$ and $D$ in the motivating example. In our ablation, sharpening helps most at $v = 0.3$ (Table~\ref{tab:ablation}). At low visibility, the anomaly vector has no clear peak, so the adaptive exponent below keeps $q$ close to~1.}

\vspace{1mm}
\noindent\textbf{Adaptive exponent}.
We set~$q$ adaptively from the signal quality $\sigma = b_{\max} / \bar{b}$ (the ratio of peak to mean belief):
\begin{equation}\label{eq:adaptive-q}
  q_{\mathrm{adapt}} = 1 + (q - 1) \cdot \min\!\Bigl(1,\;
    \max\!\bigl(0,\; \tfrac{\sigma - 5}{15}\bigr)\Bigr)\,.
\end{equation}
When the signal is weak ($\sigma < 5$), $q_{\mathrm{adapt}} = 1$ (no sharpening), which encourages broad exploration. When the signal is strong ($\sigma > 20$), the full sharpening~$q$ is applied.
We fixed the thresholds 5 and 20 before any evaluation, as round values for weak and strong signal concentration. They are the same for all topologies and visibility levels.

\vspace{1mm}\noindent\textbf{Belief refinement (optional).}\label{sec:belief}
eIRWR can wrap the inner power iteration in an outer loop that refines its belief about which nodes are suspicious. At outer iteration~$t$:
\begin{equation}\label{eq:belief}
  \bm{b}^{(t)} = (1 - \mu)\,\bm{s}_{\mathrm{obs}} + \mu\,\bm{r}^{(t-1)}\,,
\end{equation}
where $\mu \in [0, 1]$ is a momentum parameter (default $\mu = 0.1$). The belief~$\bm{b}^{(t)}$ affects both the transition matrix (through~$\bm{C}$) and the teleportation vector (through Eq.~\eqref{eq:teleport}). A small $\mu$ keeps the belief close to the raw observation. This avoids confirmation bias, where the walk reinforces its own prior. \rev{With more than one outer iteration, once a pass moves probability from the victim $B$ toward $D$, the refined belief raises $D$. The next pass then uses self-loops and a restart vector that favor $D$ over $B$.}

The outer loop stops when $\|\bm{r}^{(t)} - \bm{r}^{(t-1)}\|_1 < \varepsilon_{\mathrm{outer}}$. \rev{With a single outer iteration ($n_{\mathrm{outer}} = 1$), the belief equals the observation ($\bm{b} = \bm{s}_{\mathrm{obs}}$), so refinement is inactive. All experiments in this paper use $n_{\mathrm{outer}} = 1$, so the reported results do not depend on refinement.}

Algorithm~\ref{alg:irwre} summarizes the complete eIRWR procedure.

\begin{algorithm}[t]
\caption{eIRWR: Sharpened Anomaly-Restart Random Walk}
\label{alg:irwre}
\footnotesize
\begin{algorithmic}[1]
\Require $\bm{W}$, $\bm{s}_{\mathrm{obs}}$, $R_{\mathrm{base}}, \beta, \alpha, \rho, q, \mu, \varepsilon, \varepsilon_{\mathrm{outer}}, n_{\mathrm{outer}}$
\Ensure Root cause probability vector $\bm{r}$
\State $\hat{\bm{s}} \gets \bm{s}_{\mathrm{obs}} / \max(\bm{s}_{\mathrm{obs}})$;\; $\bm{R} \gets R_{\mathrm{base}} \exp(-\beta\,\hat{\bm{s}})$
\State $\bm{M}_{\mathrm{base}} \gets \mathrm{diag}(\bm{1} - \bm{R})\,\bm{W}$;\; $\bm{r} \gets \bm{s}_{\mathrm{obs}}$
\For{$t = 1$ \textbf{to} $n_{\mathrm{outer}}$}
  \State $\bm{b} \gets (1-\mu)\,\bm{s}_{\mathrm{obs}} + \mu\,\bm{r}$;\; $\bm{C} \gets \bm{b} / \max(\bm{b})$
  \State $\bm{M} \gets \mathrm{RowNorm}(\bm{M}_{\mathrm{base}}\,\mathrm{diag}(\bm{C}) + \bm{A}_{\mathrm{bwd}} + \bm{A}_{\mathrm{self}})$
  \State $\sigma \gets \max(\bm{b}) / \bar{b}$;\; $q_{\mathrm{adapt}} \gets$ Eq.~\eqref{eq:adaptive-q}
  \State $\bm{v} \gets \bm{b}^{\,q_{\mathrm{adapt}}} / \|\bm{b}^{\,q_{\mathrm{adapt}}}\|_1$;\; $\bm{r}_{\mathrm{old}} \gets \bm{r}$
  \Repeat\ $\bm{r} \gets (1-\alpha)\,\bm{M}\,\bm{r} + \alpha\,\bm{v}$ \Until{$\|\Delta\bm{r}\|_1 < \varepsilon$}
  \State \textbf{if} $t > 1$ \textbf{and} $\|\bm{r} - \bm{r}_{\mathrm{old}}\|_1 < \varepsilon_{\mathrm{outer}}$ \textbf{then break}
\EndFor
\State \Return $\bm{r}$;\; root cause $v^* = \arg\max_i\, r_i$
\end{algorithmic}
\end{algorithm}

\vspace{1mm}\noindent\textbf{Complexity analysis.}\label{sec:complexity}
Each inner iteration does one sparse matrix-vector product ($O(|E|)$) plus $O(N)$ work for the teleportation term. Building~$\bm{M}$ takes $O(|E|)$ for the forward transitions, $O(|E|)$ for the backward edges, and $O(N)$ for self-loops and row normalization. The total cost per incident is $O\bigl(n_{\mathrm{outer}} \cdot (K_{\mathrm{inner}} \cdot |E| + |E|)\bigr)$, where $K_{\mathrm{inner}}$ is the number of power-iteration steps (\rev{about 60 at the default $\alpha = 0.15$; Fig.~\ref{fig:convergence}}). All matrices are stored in CSR format, so memory is $O(|E| + N)$.

\vspace{1mm}\noindent\textbf{Comparison with PPR.}\label{sec:vs-ppr}
Basic IRWR and the PPR at the core of MicroRCA use the same power iteration. They differ only in the resilience damping of Eq.~\eqref{eq:fault-prop}:
\begin{align}
  \text{PPR:}  \quad &\bm{r}^{(k+1)} = (1-\alpha)\,\bm{W}\,\bm{r}^{(k)}
    + \alpha\,\bm{s}_{\mathrm{obs}} \label{eq:ppr} \\
  \text{IRWR:} \quad &\bm{r}^{(k+1)} = (1-\alpha)\,\bm{M}_R\,\bm{r}^{(k)}
    + \alpha\,\bm{s}_{\mathrm{obs}}, \label{eq:irwr-vs}
\end{align}
Both walk on the same matrix $\bm{W}$ in the same direction. The only change is the multiplicative damping $\mathrm{diag}(\bm 1-\bm R)$. This damping may look like a new ``fault-propagation'' model. With uniform resilience, however, it has a much simpler interpretation.

\begin{proposition}[Resilience damping is an effective restart]
\label{prop:effective-restart}
    Let $R_i \equiv R$ be uniform, so $\bm{M}_R = (1-R)\,\bm{W}$. Then the converged IRWR score with restart $\alpha$ induces the same node ranking as PPR with restart 
    \begin{equation}\label{eq:alpha-eff}
          \alpha_{\mathrm{eff}} = 1 - (1-\alpha)(1-R) = \alpha + R - \alpha R .
    \end{equation}
\end{proposition}

\begin{proof}
    At convergence, Eq.~\eqref{eq:irwr-vs} gives $\bm r = \alpha\bigl(\bm I-(1-\alpha)(1-R)\bm W\bigr)^{-1}\bm s_{\mathrm{obs}}$, and PPR with restart $\alpha_{\mathrm{eff}}$ gives $\bm r' = \alpha_{\mathrm{eff}}\bigl(\bm I-(1-\alpha_{\mathrm{eff}})\bm W\bigr)^{-1}\bm s_{\mathrm{obs}}$.  Choosing $\alpha_{\mathrm{eff}}$ as in Eq.~\eqref{eq:alpha-eff} makes $(1-\alpha_{\mathrm{eff}})=(1-\alpha)(1-R)$, so $\bm r$ and $\bm r'$ are equal up to the positive scalar $\alpha/\alpha_{\mathrm{eff}}$, which does not change the ranking.
\end{proof}

\begin{remark}\label{rem:alpha-eff}
    For this paper's default $\alpha=0.15,\,R=0.1$, Eq.~\eqref{eq:alpha-eff} gives $\alpha_{\mathrm{eff}} = 1-(1-0.15)(1-0.1) = 0.235$: basic IRWR is exactly PPR at restart $0.235$.
\end{remark}

Proposition~\ref{prop:effective-restart} shows that basic IRWR is not a new graph model. The ablation in Section~\ref{sec:ablation} confirms this equality numerically. It also shows that walking on the true transpose $\bm{W}^\top$ (a genuinely ``backward'' matrix) is catastrophic. Basic IRWR must therefore be compared with a restart-tuned PPR, not the default one. Against that baseline, its advantage disappears. The large gains that remain come from the restart-shaping and structural mechanisms of eIRWR.

On the example graph of Fig.~\ref{fig:irwr-overview}, this is the difference between a naive anomaly restart, which piles probability on \rev{the most anomalous victim}~$B$, and eIRWR, \rev{whose backward edges and self-loop move} probability toward the originating fault~$D$.

\vspace{1mm}\noindent\textbf{Example.}\label{sec:worked-example}
\rev{We now follow the motivating example of Fig.~\ref{fig:irwr-overview} (Section~\ref{sec:introduction}) through the mechanisms above. All scores are computed on the full CG84 graph at visibility $v = 0.3$. The raw anomaly vector restarts the walk most often at the victim~$B$, which has the highest anomaly score of all services. Basic IRWR ranks $D$ 15th, and PPR ranks it 86th. In these walks, a service collects the scores of its callees. $B$ and $C$ therefore collect the score of $D$, while $D$, which calls no other service, keeps only its restart share.

eIRWR ranks $D$ first, with $r_B = 0.83\,r_D$ (Fig.~\ref{fig:irwr-overview}, right). Backward edges decide this result. In the row of $D$ in $\bm{M}$, backward edges send $46\%$ of the weight to the callers of $D$, and the self-loop keeps the rest at $D$. These callers, $B$ and $C$, are direct victims and are strongly anomalous. In the row of $B$, backward edges send $37\%$ of the weight to the callers of $B$, which are barely anomalous because only a small share of their calls reaches $B$. Without backward edges ($\rho = 0$), $B$ is ranked first and $D$ second. Without self-loops, without sharpening ($q = 1$), or with uniform resilience ($\beta = 0$), $D$ stays first. Sharpening is partly active here (peak-to-mean ratio $9.6$, $q_{\mathrm{adapt}} = 1.31$). Over 200 noise draws of the same incident, eIRWR ranks $D$ first in all of them. We chose this incident because a victim has the highest anomaly score. It illustrates the mechanisms but does not show how often eIRWR succeeds, which Section~\ref{sec:main-results} reports over all incidents. A notebook in our repository reproduces these numbers.}

\section{Experimental Results}\label{sec:experiments}

\subsection{Experiment Setup}\label{sec:setup}

\noindent\textbf{Dataset.}
The Alibaba Microservice Trace Dataset (v2022)~\cite{alibaba2022trace} provides distributed traces from a production cluster. We select three representative call-graph snapshots of increasing scale (Table~\ref{tab:dataset-stats}).

\begin{table}[t]
\centering
\caption{Dataset characteristics of the three topologies, counted over all trace rows. The dependency graphs used in the evaluation keep only calls with a recorded response time and have 11,966, 17,379, and 24,368 nodes.}
\label{tab:dataset-stats}

\begin{tabular*}{\columnwidth}{@{\extracolsep{\fill}}lrrrr@{}}
    \toprule
    \textbf{Topology} & \textbf{Trace Rows} & \textbf{Nodes} & \textbf{Edges} & \textbf{Avg Deg.} \\
    \midrule
    CG84  &  8,469,061 & 12,453 & 30,742 & 2.47 \\
    CG136 & 15,415,374 & 17,842 & 56,185 & 3.15 \\
    CG297 & 20,605,953 & 24,952 & 87,914 & 3.52 \\
    \bottomrule
\end{tabular*}
\end{table}

\begin{figure*}
    \centering
    \includegraphics[width=.85\linewidth]{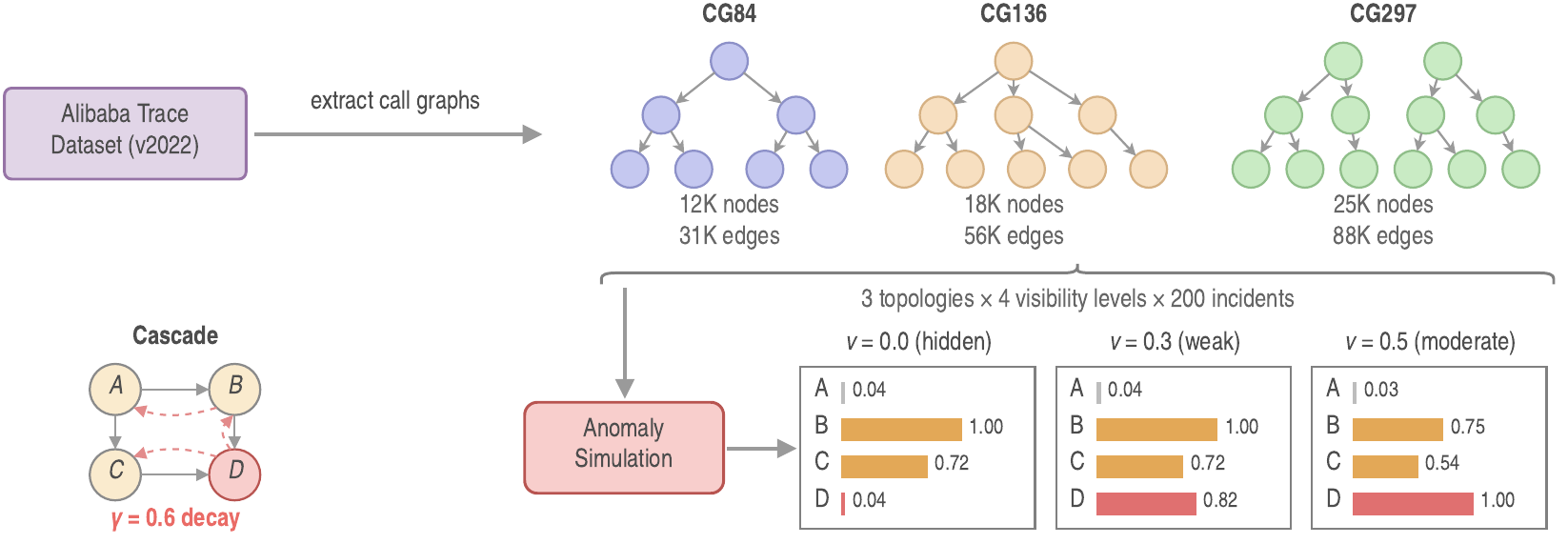} 
    \caption{\rev{A qualitative illustration of our evaluation pipeline.} Traces from the Alibaba production cluster are aggregated into three call-graph topologies of increasing scale (top). For each topology, we simulate incidents by selecting a root cause and propagating faults backward with geometric decay (bottom). The visibility parameter~$v$ controls how much anomaly signal the root cause~$D$ contributes to the observed vector~$\bm{s}_{\mathrm{obs}}$. At $v\!=\!0$, cascading symptoms fully mask the root cause. At $v\!=\!0.5$, it \rev{has the highest anomaly score. The bars show the incident of Fig.~\ref{fig:irwr-overview}, with each score divided by the largest score among all services.}
    }
    \label{fig:eval-pipeline}
\end{figure*}

\vspace{1mm}\noindent
\textbf{Anomaly simulation.}
For each topology, we simulate incidents by selecting a random callee node as the root cause and propagating a fault signal backward through the dependency graph. \rev{The fault signal is~1 at the root cause. A caller~$i$ of an affected service~$j$ receives $\gamma\,w_{ij}$ times the signal of~$j$, so it inherits the share of its calls that go to~$j$.} Propagation uses geometric decay ($\gamma = 0.6$) over 5~hops. \rev{The root cause itself receives the anomaly score~$v$\revtwo{, plus any signal that returns to it through self-calls or call cycles (in 41, 45, and 61 of the 200 incidents on CG84, CG136, and CG297)}. We then add half-normal noise $|\mathcal{N}(0, \sigma^2)|$ with $\sigma = 0.05$ to every service and normalize the vector to sum to one.}
The root-cause visibility~$v$ controls how visible the root cause is in the observed anomaly vector. \rev{We use seven levels, $v \in \{0, 0.05, 0.1, 0.15, 0.2, 0.3, 0.5\}$. At $v = 0$, the root cause shows only noise \revtwo{and this returned signal, so cascading symptoms mask it in most incidents}. The propagated symptoms and the noise do not depend on~$v$, so a larger~$v$ makes the root cause stand out more clearly.} We run 200~incidents per topology at each visibility level. Every method sees the same incidents (identical random seeds), which gives the paired design assumed by the significance tests in Section~\ref{sec:main-results}. Fig.~\ref{fig:eval-pipeline} shows the evaluation pipeline and the effect of~$v$ on the observed anomaly vector.

\vspace{1mm}\noindent
\textbf{Baselines.}
We compare against eight methods that, like ours, use aggregate anomaly vectors and the service-level call graph: MicroRCA~\cite{Wu:2020:NOMS} (PPR), CloudRanger~\cite{Wang:2018:CCGRID} (\revtwo{anomaly-weighted random walk, implemented as a first-order walk}), MonitorRank~\cite{Kim:2013:SIGMETRICS} (anomaly-correlated walk), Corr-Fusion (a correlation-fusion composite we construct in the spirit of MonitorRank: an anomaly-teleport walk on \revtwo{$\lambda\,\bm{W} + (1-\lambda)\,\bm{C}$ with $\lambda = 0.5$}, where~$\bm{C}$ scales each edge by the anomaly scores of both endpoints), In-Degree Centrality \revtwo{(sum of the call shares a service receives)}, MicroHECL~\cite{Liu:2021:ICSE-SEIP} (\revtwo{five-hop traversal in which each service collects the anomaly of the services it calls, without its own score}), TraceDiag~\cite{Ding:2023:ACM} (\revtwo{pruning below the median anomaly score, then PPR}), and Neighbor Correlation (PAL/FChain-style~\cite{pal2020rca, fchain2022}). \rev{eIRWR needs the call graph, so we compare it only with methods that also use the graph. Methods that need no graph, such as ranking services by anomaly score, are outside the scope of this comparison. They should be used when no call graph is available.}

We exclude trace-level methods (MicroRank~\cite{yu2021microrank}, TraceRank~\cite{Yu:2023:J.Softw.Evol.Process}) because they require per-request span data, a different and richer input. Under simulation, this difference is decisive. The per-trace pass/fail labels these methods use are derived directly from the injected fault, so they act as an oracle-informed upper bound. For example, MicroRank attains MRR~$\approx 0.9$ even at zero root-cause visibility, where no aggregate-metric method exceeds \rev{MRR~$= 0.08$}. Including them would conflate input richness with algorithmic merit. We therefore treat them as an upper bound that uses a strictly richer input.

\vspace{1mm}\noindent
\textbf{Metrics.}
We report Precision at Rank~$k$ (PR@$k$, for $k \in \{1, 3, 5\}$), MRR, and the average rank of the true root cause.
All metrics are averaged over incidents within each topology, then averaged across the three topologies.

\vspace{1mm}\noindent
\textbf{Parameters.}
Unless stated otherwise, we use restart probability $\alpha = 0.15$, base resilience $R_{\mathrm{base}} = 0.1$, anomaly sensitivity $\beta = 2.0$, backward weight $\rho = 0.3$, teleportation exponent $q = 2.0$, momentum $\mu = 0.1$, convergence tolerance $\varepsilon = 10^{-6}$, \rev{a cap of 500 inner power iterations for eIRWR (1{,}000 for IRWR),} and \rev{$n_{\mathrm{outer}} = 1$ (refinement inactive)}.
We fixed these values during development on small synthetic graphs, before processing the evaluation topologies, and did not tune them on the evaluation incidents. The same values are used for all topologies and visibility levels.

\subsection{Results}
\label{sec:main-results}

\begin{table*}[htb]
\centering
\caption{RCA performance averaged over three topologies (CG84, CG136, CG297). $v$~denotes root-cause visibility. \rev{Bold marks the best value in each column.}}
\label{tab:main-results}

\setlength{\tabcolsep}{3pt}
\begin{tabular*}{\textwidth}{@{\extracolsep{\fill}}l|rrr|rrr|rrr|rrr@{}}
    \toprule
    & \multicolumn{3}{c|}{$v = 0.0$}
    & \multicolumn{3}{c|}{$v = 0.1$}
    & \multicolumn{3}{c|}{$v = 0.3$}
    & \multicolumn{3}{c}{$v = 0.5$} \\
    \textbf{Method} & PR@1 & PR@5 & MRR
                    & PR@1 & PR@5 & MRR
                    & PR@1 & PR@5 & MRR
                    & PR@1 & PR@5 & MRR \\
    \midrule
    MicroRCA      & 0.02 & 0.03 & 0.02 & 0.02 & 0.03 & 0.03 & 0.03 & 0.06 & 0.06 & 0.06 & 0.20 & 0.17 \\
    CloudRanger   & 0.00 & 0.01 & 0.01 & 0.00 & 0.02 & 0.01 & 0.00 & 0.02 & 0.02 & 0.01 & 0.02 & 0.02 \\
    MonitorRank   & 0.02 & 0.03 & 0.02 & 0.03 & 0.04 & 0.04 & 0.06 & 0.09 & 0.09 & 0.10 & 0.22 & 0.19 \\
    Corr-Fusion   & 0.02 & 0.03 & 0.02 & 0.02 & 0.03 & 0.03 & 0.04 & 0.07 & 0.07 & 0.07 & 0.21 & 0.18 \\
    In-Degree     & 0.00 & 0.00 & 0.00 & 0.00 & 0.00 & 0.00 & 0.00 & 0.00 & 0.00 & 0.00 & 0.00 & 0.00 \\
    MicroHECL     & 0.01 & 0.02 & 0.01 & 0.02 & 0.02 & 0.02 & 0.02 & 0.03 & 0.02 & 0.02 & 0.03 & 0.02 \\
    TraceDiag     & 0.02 & 0.03 & 0.02 & 0.02 & 0.04 & 0.03 & 0.04 & 0.06 & 0.06 & 0.07 & 0.13 & 0.12 \\
    Neigh.\ Corr. & 0.00 & 0.01 & 0.01 & 0.00 & 0.01 & 0.01 & 0.00 & 0.01 & 0.02 & 0.00 & 0.01 & 0.02 \\
    \midrule
    IRWR          & 0.02 & 0.03 & 0.02 & 0.03 & 0.05 & 0.04 & 0.07 & 0.23 & 0.18 & 0.46 & 0.95 & 0.66 \\
    eIRWR         & \textbf{0.03} & \textbf{0.14} & \textbf{0.08} & \textbf{0.07} & \textbf{0.24} & \textbf{0.15} & \textbf{0.65} & \textbf{0.95} & \textbf{0.78} & \textbf{0.91} & \textbf{0.99} & \textbf{0.95} \\
    \bottomrule
\end{tabular*}
\end{table*}

\begin{figure}[htb]
\centering
    \includegraphics[width=\linewidth]{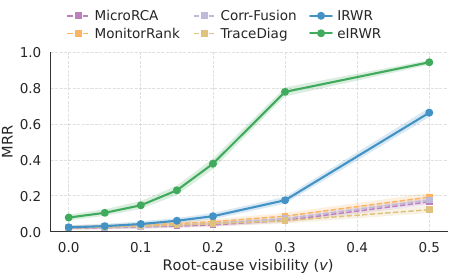}
    \caption{MRR vs root-cause visibility, averaged over three topologies. eIRWR has a higher MRR than every baseline at every visibility level. \rev{Shaded bands show 95\% bootstrap confidence intervals over incidents.}}
    \label{fig:mrr-visibility}
\end{figure}

Table~\ref{tab:main-results} gives the main comparison, and Fig.~\ref{fig:mrr-visibility} shows MRR against root-cause visibility.

\rev{eIRWR has the highest PR@1, PR@5, and MRR of all compared methods at every visibility level. At $v = 0.3$, it achieves MRR~$= 0.78$ and PR@1~$= 0.65$. The best baseline (MonitorRank) reaches MRR~$= 0.09$ and PR@1~$= 0.06$, and basic IRWR reaches MRR~$= 0.18$. At $v = 0.5$, eIRWR reaches MRR~$= 0.95$ and ranks the root cause first in $91\%$ of incidents. At low visibility, every method is weak, but eIRWR still has the highest MRR ($0.08$ at $v = 0$ and $0.15$ at $v = 0.1$, against at most \revtwo{$0.02$} and $0.04$ for the baselines). \revtwo{At $v = 0$, the baselines find the root cause only in incidents where signal returns to it. In the other incidents, their MRR is at most $0.01$, and eIRWR reaches $0.05$.}}

The enhancements, not the \rev{resilience-damped} transition matrix, drive the gain. Basic IRWR improves on MicroRCA at $v=0.3$ (MRR \rev{0.18 vs.\ 0.06}), but the ablation below shows that this is entirely an effective-restart effect. eIRWR's \rev{enhancements} then lift MRR from \rev{0.18 to 0.78, PR@5 from 0.23 to 0.95, and PR@1 from 0.07 to 0.65}. At $v=0.1$, eIRWR reduces the average rank of the true root cause from \rev{920 to 316}.

Structural and sampling baselines collapse at scale. In-Degree Centrality, CloudRanger, \rev{MicroHECL,} and Neighbor Correlation have near-zero MRR ($\leq \rev{0.03}$) at every visibility level. \rev{The PPR-based baselines (MicroRCA, MonitorRank, Corr-Fusion, and TraceDiag) stay below MRR~$= 0.2$ even at $v = 0.5$.} \revtwo{Corr-Fusion behaves almost like MicroRCA. Its matrix~$\bm{C}$ scales each edge of~$\bm{W}$ by the normalized anomaly scores of both endpoints, which are small for most services (median edge factor about $0.008$ on CG84 at $v = 0.3$). The blended matrix therefore stays close to~$\bm{W}$.} We analyze these failure modes in Section~\ref{sec:discussion}.

\vspace{1mm}\noindent \textbf{Ablation study.}\label{sec:ablation}
To attribute the improvement correctly, we split our method into the components claimed by prior work and by us. We evaluate each variant on \revtwo{the CG84 incidents of the main comparison} (200 per visibility, seed~42), and Table~\ref{tab:ablation} reports the results. A single seed on one topology suffices, because the ablation is diagnostic rather than comparative. Its second conclusion is an exact algebraic identity (Proposition~\ref{prop:effective-restart}) that holds on any graph, and Fig.~\ref{fig:per-topology} confirms that the headline gain transfers to CG136 and CG297. Three conclusions follow.

\begin{table}[t]
\centering
\caption{Ablation on CG84 (MRR, 200 incidents per level, seed 42, $\alpha=0.15$ unless noted). $\bm{M}_R$: resilience-damped matrix. \rev{Indented rows switch off parts of eIRWR. ``None of the three'' switches off backward edges, self-loops, and sharpening, which leaves only the belief-weighted forward transitions of Eq.~\eqref{eq:M-combined}.}}
\label{tab:ablation}

\begin{tabular*}{\columnwidth}{@{\extracolsep{\fill}}lrrr@{}}
    \toprule
    \textbf{Variant} & \textbf{$v{=}0.1$} & \textbf{$v{=}0.3$} & \textbf{$v{=}0.5$} \\
    \midrule
    MicroRCA: PPR on $\bm{W}$          & \revtwo{0.04} & \revtwo{0.08} & \revtwo{0.18} \\
    PPR on true transpose $\bm{W}^\top$ & \revtwo{0.02} & 0.02 & \revtwo{0.03} \\
    Basic IRWR on $\bm{M}_R$ ($R{=}0.1$) & \revtwo{0.06} & 0.18 & \revtwo{0.66} \\
    PPR on $\bm{W}$, tuned $\alpha{=}0.235$ & \revtwo{0.06} & 0.18 & \revtwo{0.66} \\
    \midrule
    \textbf{eIRWR (full)}            & 0.17 & \revtwo{0.79} & \revtwo{0.93} \\
    \rev{\quad no backward edges ($\rho{=}0$)} & \rev{0.12} & \rev{0.87} & \revtwo{0.92} \\
    \rev{\quad no self-loops} & \rev{0.15} & \revtwo{0.54} & \revtwo{0.91} \\
    \rev{\quad no sharpening ($q{=}1$)} & \rev{0.16} & \revtwo{0.76} & \revtwo{0.93} \\
    \rev{\quad none of the three} & \revtwo{0.06} & \rev{0.11} & \rev{0.21} \\
    \rev{\quad uniform resilience ($\beta{=}0$)} & \rev{0.17} & \revtwo{0.79} & \revtwo{0.93} \\
    \bottomrule
\end{tabular*}
\end{table}

\rev{Reversing the walk} is not the mechanism.
Walking on the true transpose $\bm{W}^\top$, which moves probability from caller to callee and pushes mass toward leaf services, is catastrophic (MRR \revtwo{$\leq 0.03$}). MicroRCA's PPR on $\bm{W}$ already adds callee scores to callers, so there is no other ``backward'' matrix to use. \rev{The backward edges of eIRWR are different. They are weak, weighted by the belief of the callee, and added to the forward transitions rather than replacing them.}

Resilience damping is an effective restart. Basic IRWR ($\bm{M}_R$, $R=0.1$) and PPR with $\alpha_{\mathrm{eff}}=0.235$ produce the same MRR at every visibility (\revtwo{$0.06$, $0.18$, and $0.66$}), as Proposition~\ref{prop:effective-restart} predicts. \rev{The unrounded values differ by less than $10^{-4}$, because the finite convergence tolerance breaks a few near-ties differently.} The whole gain of IRWR over MicroRCA is a restart-rate effect, which a practitioner gets for free by tuning $\alpha$.

\rev{The remaining gain comes from the mechanisms of eIRWR, which overlap. Against the restart-tuned baseline ($0.18$ at $v = 0.3$), eIRWR raises MRR to \revtwo{$0.79$}. Switching off backward edges, self-loops, and sharpening together removes this gain ($0.11$). Switching off any one of them keeps most of it, so each mechanism can partly replace the others. Self-loops contribute most at $v = 0.3$ (\revtwo{$0.54$} without them), and sharpening adds a smaller gain (\revtwo{$0.76$} without it). Backward edges help at $v = 0.1$ ($0.12$ without them) but lower MRR at $v = 0.3$ ($0.87$ without them). A likely reason is that, when the root cause has a high anomaly score, backward edges let it pass part of its probability to its callers. A moderately anomalous service whose neighbors are quiet keeps almost all of its probability through its self-loop and can then outrank the root cause.
Anomaly-conditioned resilience (Section~\ref{sec:adaptive-R}) contributes little on its own: varying $\beta$ over $[0,10]$ changes MRR at $v=0.3$ by less than $0.01$ (\revtwo{from $0.794$ to $0.796$}).}


\vspace{1mm}\noindent\textbf{Cross-topology generalization.}
\rev{At $v = 0.3$, the accuracy of eIRWR changes little as the graph grows from 12K to 25K nodes. Fig.~\ref{fig:per-topology} gives the per-topology breakdown. eIRWR reaches MRR~$= 0.79$, $0.80$, and $0.75$ on CG84, CG136, and CG297. Basic IRWR stays between $0.16$ and $0.18$, and the best baseline never exceeds MRR~$= 0.10$. At low visibility, every method is less accurate on the largest graph. At $v = 0.1$, for example, eIRWR reaches MRR~$= 0.17$, $0.17$, and $0.10$.} CloudRanger, In-Degree, MicroHECL, and Neighbor Correlation are near zero on all three graphs, so we omit them from the figure.

\begin{figure}[htb]
\centering
    \includegraphics[width=\linewidth]{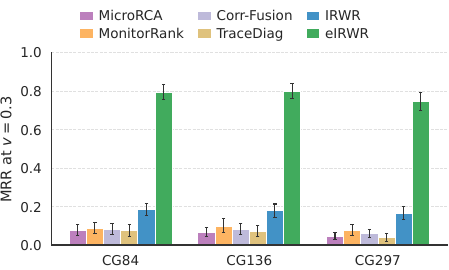}
    \caption{MRR by method across three topologies at $v = 0.3$. \rev{eIRWR has the highest MRR on all three graphs. Error bars show 95\% bootstrap confidence intervals over 200 incidents.}}
    \label{fig:per-topology}
\end{figure}


\begin{table}[t]
\centering
\caption{\rev{Two-sided} Wilcoxon signed-rank tests \rev{on the rank of the root cause}, eIRWR vs.\ each baseline at $v = 0.3$ \rev{(200 incidents on CG297)}. $r = |Z|/\sqrt{N}$: effect size, with $N$ the number of non-tied pairs. \rev{In every row, eIRWR ranks the root cause higher than the baseline.} \revtwo{W/T/L: incidents in which eIRWR ranks the root cause higher, equally, or lower. $r$ is at most about $0.87$, reached with no losses.}}
\label{tab:wilcoxon}
\begin{tabular*}{\columnwidth}{@{\extracolsep{\fill}}lrrr@{}}
    \toprule
    \textbf{Baseline}  & \textbf{$p$-value} & \textbf{$r$} & \revtwo{\textbf{W/T/L}} \\
    \midrule
    MicroRCA      & $4.9\times10^{-34}$ & 0.87 & \revtwo{196/3/1} \\
    CloudRanger   & $2.1\times10^{-34}$ & 0.87 & \revtwo{198/1/1} \\
    MonitorRank   & $9.1\times10^{-33}$ & 0.87 & \revtwo{189/11/0} \\
    Corr-Fusion   & $1.1\times10^{-33}$ & 0.87 & \revtwo{194/5/1} \\
    In-Degree     & $1.4\times10^{-34}$ & 0.87 & \revtwo{200/0/0} \\
    MicroHECL     & $2.1\times10^{-34}$ & 0.87 & \revtwo{199/1/0} \\
    TraceDiag     & $6.5\times10^{-34}$ & 0.87 & \revtwo{196/4/0} \\
    Neigh.\ Corr. & $1.5\times10^{-34}$ & 0.87 & \revtwo{199/0/1} \\
    \midrule
    IRWR          & $2.0\times10^{-28}$ & 0.82 & \revtwo{176/16/8} \\
    \bottomrule
\end{tabular*}
\end{table}

The Friedman test rejects the equality of \rev{the compared} methods at every visibility level ($p < 10^{-15}$). \rev{Two-sided} paired Wilcoxon signed-rank tests \rev{(Table~\ref{tab:wilcoxon}, 200 incidents on CG297) show that eIRWR ranks the root cause higher than every baseline at $v = 0.3$ ($p < 10^{-27}$), with large effect sizes (\revtwo{$r > 0.8$}). All nine tests remain significant after Bonferroni correction ($\alpha = 0.05/9$).}

\vspace{1mm}\noindent\textbf{Computational efficiency.}\label{sec:efficiency}
Fig.~\ref{fig:scalability} shows runtime and peak memory against graph size. Table~\ref{tab:computational} reports values at the largest point of the runtime sweep (17,468~nodes, 55,125~edges).
We build the sweep from CG297, the largest topology, by rebuilding the dependency graph from growing prefixes of its call-graph trace ($10^{5}$ to $5\times10^{6}$ records). This gives graphs of 3,971 to 17,468 nodes. The largest point is smaller than the full CG297 dependency graph (24,368 nodes) because node count grows sub-linearly with the number of records: later records mostly revisit services already in the graph. We do not sweep CG84 and CG136 separately. The sweep aims to isolate the effect of graph size within one topology, and their \revtwo{dependency graphs (11,966 and 17,379 nodes) lie within the swept range}.

\begin{figure}[t]
\centering
    \includegraphics[width=\linewidth]{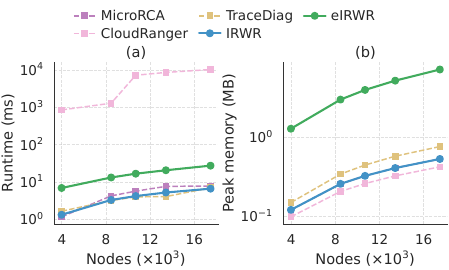}
    \caption{(a)~Runtime vs. graph size and (b)~peak memory vs.\ graph size. IRWR costs about the same as MicroRCA. eIRWR is \rev{about four times slower than IRWR} because it builds backward transitions and self-loops, but stays under \rev{30}\,ms even at 17K~nodes.}
    \label{fig:scalability}
\end{figure}

\begin{table}[t]
\centering
\caption{Computational cost at the largest scale of the runtime sweep ($N = 17,468$, $|E| = 55,125$).}
\label{tab:computational}
\begin{tabular*}{\columnwidth}{@{\extracolsep{\fill}}lrr@{}}
    \toprule
    \textbf{Method} & \textbf{Time (ms)} & \textbf{Memory (MB)} \\
    \midrule
    MicroRCA        & 7.8      & 0.5 \\
    CloudRanger     & 10,372.9 & 0.4 \\
    TraceDiag       & 7.1      & 0.8 \\
    \midrule
    IRWR            & 6.7      & 0.5 \\
    eIRWR           & 27.6     & 7.2 \\
    \bottomrule
\end{tabular*}
\end{table}

IRWR runs in \rev{6.7}\,ms, comparable to MicroRCA (\rev{7.8}\,ms) and TraceDiag (\rev{7.1}\,ms). This is expected, because all three use sparse power iteration. IRWR and MicroRCA differ only in the transition matrix ($\bm{M}_R$ vs.\ $\bm{W}$). eIRWR takes \rev{27.6}\,ms because it also builds backward transitions and self-loops. This is still well within operational latency budgets for online RCA. CloudRanger is orders of magnitude slower (${\sim}\rev{10}$\,s). \revtwo{It samples 50,000 walks of length 20 per incident (10,000 on sweep graphs with at most 10,000 nodes).}
Fig.~\ref{fig:convergence} shows the number of iterations to convergence against the restart probability~$\alpha$. At the default $\alpha = 0.15$, IRWR converges in ${\sim}28$ iterations and eIRWR in ${\sim}\rev{60}$. Higher~$\alpha$ gives faster convergence but shallower exploration. \rev{At the smallest value, $\alpha = 0.01$, eIRWR stops at its cap of 500 iterations in all 20 incidents.}

\begin{figure}[t]
    \centering
    \includegraphics[width=\linewidth]{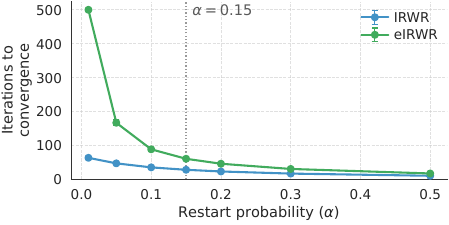}
    \caption{Iterations to convergence vs.\ restart probability~$\alpha$ for IRWR and eIRWR, on a subgraph from the scalability sweep of Fig.~\ref{fig:scalability} ($N = 10,688$). Error bars show $\pm 1$~standard deviation over 20 incidents.}
    \label{fig:convergence}
\end{figure}

\vspace{1mm}\noindent\textbf{Real-anomaly-window case study.}\label{sec:realworld}
The controlled evaluation uses simulated incidents so that ground-truth labels exist. \revtwo{As a check on real data, we use the Alibaba response-time records for the three minutes covered by CG297. Each one-minute snapshot is one window, and $\bm{s}_{\mathrm{obs}}$ holds the positive z-scores of the mean RPC response time of each service. In each window, 679 to 700 services have at least one instance reading more than three standard deviations above their mean. Without labels, we score each top-1 pick by its cascade coverage. This is the share of the observed anomaly that a failure at the pick would explain under a simplified form of our fault model (the pick and its callers within five hops, weighted by $0.6$ per hop, without call shares).
This diagnostic does not separate the methods. In each window, 2,265 to 2,594 services have a positive score, and no service holds more than $0.2\%$ of the total. Coverage therefore mostly reflects the number of upstream services within five hops (Spearman $\rho = 0.92$ to $0.94$). The most-called service, nominated by In-Degree and Neighbor Correlation, has the highest coverage ($0.19$ to $0.20$), although it is not anomalous itself. All other picks explain less than $5\%$. eIRWR picks highly anomalous services with at most two callers, which explain at most $0.15\%$. These windows neither confirm nor refute any method, and they do not show an advantage of eIRWR on real data.}

\section{Discussion}
\label{sec:discussion}

eIRWR's accuracy comes from \rev{where the walk restarts and where it lingers, not from resilience damping or from reversing the walk}. Sharpening concentrates restart mass on the most suspicious nodes, self-loops let probability accumulate at cascade sources, and backward edges allow limited upstream exploration. Together, they control where the walk restarts and lingers. This is the axis along which a root cause differs from its \rev{more anomalous} victims. \rev{The three mechanisms overlap, and any two of them keep most of the gain (Table~\ref{tab:ablation}).}

The gain over a restart-tuned PPR is largest at $v = 0.3$ (MRR \rev{$0.18 \to 0.78$}; PR@5 \rev{$0.23 \to 0.95$}). At $v = 0.0$, \revtwo{mostly} victims show symptoms, and every aggregate-metric method is limited (eIRWR reaches MRR~$= \rev{0.08}$\rev{, the best baseline \revtwo{$0.02$}}).
The baselines fail in three ways. In-Degree Centrality ignores the anomaly vector and ranks every incident identically (MRR $<0.001$). CloudRanger's\rev{~\cite{Wang:2018:CCGRID}} \revtwo{walks follow calls toward callees, so visits accumulate at heavily called services and at services that call themselves, which outrank the root cause} (MRR $<\rev{0.03}$). MicroHECL's\rev{~\cite{Liu:2021:ICSE-SEIP}} \revtwo{traversal ignores the own score of a service, so a root cause that calls no other service (39\% to 46\% of the roots) scores zero}, and its MRR stays flat at about \rev{$0.02$}. The PPR family (MicroRCA\rev{~\cite{Wu:2020:NOMS}}, MonitorRank\rev{~\cite{Kim:2013:SIGMETRICS}}, Corr-Fusion) is the strongest \rev{graph-based family, but its MRR stays below $0.2$. In these walks, a service collects the scores of the services it calls. The callers of the root cause, which are its direct victims, therefore collect its score on top of their own, while a root cause that calls no other service, such as a database, keeps only its restart share (Section~\ref{sec:worked-example}). Raising the restart rate weakens this bias. eIRWR addresses it more directly: self-loops keep probability at suspicious services, and backward edges let a callee collect evidence from its anomalous callers.}

\rev{
\vspace{1mm}\noindent\textbf{Comparison with related work.}\label{sec:disc-related}
The original evaluations of the baselines used different systems, fault types, and metrics~\cite{Wu:2020:NOMS, Wang:2018:CCGRID, Kim:2013:SIGMETRICS, Liu:2021:ICSE-SEIP, Ding:2023:ACM}, so their reported accuracies are not comparable with ours. We therefore re-ran all methods on the same incidents and graphs (Section~\ref{sec:setup}). Earlier work observed that ranking services by anomaly score often blames downstream victims~\cite{soldani2022survey}, as in the example of Fig.~\ref{fig:irwr-overview}. Our results extend MicroRCA~\cite{Wu:2020:NOMS}, which uses the anomaly vector directly as the restart vector. We show that tuning the restart rate recovers the full gain of basic IRWR (Table~\ref{tab:ablation}). This is consistent with the role of the restart probability and the personalization vector in PageRank and random walk with restart~\cite{haveliwala2002topic, tong2006rwr}. The further gain of eIRWR comes from self-loops, backward edges, and sharpening together. Backward and self transitions appear in earlier RCA walks~\cite{Kim:2013:SIGMETRICS, Wang:2018:CCGRID, meng2020microcause}. eIRWR differs in weighting them by the belief score~$\bm{C}$ rather than by fixed correlation scores.

\vspace{1mm}\noindent\textbf{Practical implications.}\label{sec:disc-practical}
eIRWR needs only the service call graph and aggregate metrics, with no training data or labeled incidents. At $v = 0.3$, the root cause is among the top five services in $95\%$ of incidents (PR@5), so an on-call engineer can inspect a short list. On a graph of 17,468 nodes, eIRWR runs in 27.6\,ms and uses 7.2\,MB (Table~\ref{tab:computational}). eIRWR cannot be applied without the call graph. In that case, ranking services by their anomaly score, or another method that needs no graph, should be used instead. All parameters were fixed before evaluation and shared across topologies, which suggests that eIRWR needs no per-deployment tuning, at least on similar graphs. Our implementation is publicly available (see Code availability).

\vspace{1mm}\noindent\textbf{Theoretical implications.}\label{sec:disc-theory}
Proposition~\ref{prop:effective-restart} holds on any graph, so uniform resilience damping is a change in restart rate, not a new fault-propagation model. Methods that claim gains from such damping should be compared against a restart-tuned PPR. Walking on the true transpose $\bm{W}^\top$ is harmful (Table~\ref{tab:ablation}), so reversing the walk does not separate cause from victim. Instead, where the walk lingers and where it restarts are the main factors in anomaly-restart RCA. Self-loops make the largest single contribution, followed by sharpening. Weak backward edges help when victims dominate the anomaly vector, but they can hurt when the root cause already has a high anomaly score (Table~\ref{tab:ablation}).

\vspace{1mm}\noindent\textbf{Generalization.}\label{sec:disc-generalization}
\revtwo{Two} results support generalization. eIRWR reaches MRR~$\geq 0.75$ at $v = 0.3$ on all three topologies (12K to 25K nodes) with the same parameters (Fig.~\ref{fig:per-topology}). Proposition~\ref{prop:effective-restart} does not depend on the graph. \revtwo{The real anomaly windows give no such support (Section~\ref{sec:realworld}). In addition,} all topologies come from one provider, and the incidents are simulated with backward geometric decay. The results should transfer best to systems where faults spread from callees to callers through synchronous calls. They may transfer less well to asynchronous messaging, multiple concurrent faults, or strong retry effects, which our simulation does not model.

\vspace{1mm}\noindent\textbf{Limitations.}\label{sec:disc-limitations}
First, the simulated incidents do not capture real-world effects such as partial failures and retries. Second, we compare eIRWR only with methods that use the call graph. Methods that need no graph, such as ranking services by anomaly score, are outside our comparison. Third, all aggregate-metric methods, including ours, face a ceiling at zero visibility (MRR~$\leq 0.08$). There, trace-level methods~\cite{yu2021microrank, Yu:2023:J.Softw.Evol.Process} perform better because they use per-request evidence. Fourth, eIRWR takes about 3.5 times as long as PPR on the largest graph of the runtime sweep.
}

\section{Conclusion}
\label{sec:conclusion}

\rev{We studied why anomaly-restart random walks localize root causes in microservices. Transposing the transition matrix is harmful, and uniform resilience damping only raises the restart rate. The gain comes instead from shaping the restart distribution and adding self-loops and backward edges, whose effects overlap. On three Alibaba topologies, eIRWR is the most accurate of the compared graph-based methods at every visibility level. At $v = 0.3$, it reaches MRR~$= 0.78$, about four times a restart-tuned PPR and nine times the best other baseline, and it reaches $0.95$ at $v = 0.5$. It runs in under 30\,ms on graphs of 17,000 nodes. eIRWR needs a service call graph. Without one, ranking services by anomaly score, or another method that needs no graph, should be used.

Future work includes learning per-service resilience from historical incidents, building a labeled benchmark of real incidents, modeling temporal dynamics in the restart distribution, and combining restart shaping with lightweight trace sampling to improve accuracy at low visibility. Backward edges help at low visibility but not at $v = 0.3$, so their weight could be adapted to how strongly the anomaly vector is peaked. Evaluation on other microservice systems would also test how far our results generalize.}